\DocumentMetadata{lang=en-GB}
\documentclass[a4paper]{article}
\usepackage[
unicode,
pdfauthor={Jonah Stalknecht, 金炯錄},
pdftitle={Which Functions Admit a Positive Geometry? From Branch Cuts to String Amplitudes},
pdfsubject={hep-th,math.CV},
breaklinks=true,
final
]{hyperref}
\usepackage[utf8]{inputenc}
\usepackage[T1,T2A]{fontenc}
\usepackage[unicode,final]{hyperref}
\usepackage{amsmath,amssymb,amsthm,orcidlink,cleveref,mleftright,mathtools}
\usepackage[noadjust,nosort]{cite}
\usepackage{CJKutf8}
\usepackage[final]{microtype}
\usepackage[latin,russian,french,ngerman,british]{babel}
\usepackage[osf]{newpxtext}
\usepackage[varbb,slantedGreek]{newpxmath}
\renewcommand\pi\piup
\renewcommand\Gamma\Gammaup
\renewcommand\gamma\gammaup
\renewcommand\Re{\operatorname{Re}}
\renewcommand\Im{\operatorname{Im}}
\usepackage{authblk}

\newtheorem{theorem}{Theorem}
\newtheorem{lemma}{Lemma}
\theoremstyle{definition}
\newtheorem{definition}{Definition}
\theoremstyle{remark}
\newtheorem{example}{Example}
\title{Which Functions Admit a Positive Geometry? From Branch Cuts to String Amplitudes}
\author[1]{Hyungrok Kim}
\author[2]{Jonah Stalknecht}
\affil[1]{Centre for Mathematics and Theoretical Physics Research,\authorcr Department of Physics, Astronomy and Mathematics, \authorcr University of Hertfordshire, Hatfield, UK\authorcr\href{mailto:h.kim2@herts.ac.uk}{\texttt{h.kim2@herts.ac.uk}}}
\affil[2]{Institute for Particle and Nuclear Physics, \authorcr Charles University, Prague, Czech Republic\authorcr\href{mailto:jonah.stalknecht@matfyz.cuni.cz}{\texttt{jonah.stalknecht@matfyz.cuni.cz}}}

\begin{document}
	
\maketitle

\begin{abstract}
	Positive geometry provides a geometric framework where physical observables are encoded as canonical forms associated to regions of kinematic space. In this paper we consider a generalisation to an infinite union of line segments, which allows us to capture canonical forms beyond rational functions. In the continuum limit of positive geometries, we show that we can generalise even further and describe positive geometries whose canonical forms contain branch cuts. We will constrain which functions can be obtained as the canonical form of one-dimensional positive geometries. We introduce the notion of the \emph{pseudogenus} to classify meromorphic function, and show that canonical forms can be written as the $\mathrm{d}\log$ of a function with pseudogenus zero. Furthermore, we argue that
	the spectrum encoded by a union of line segments is consistent with the presence of a stringy tower of states or a Kaluza--Klein tower with four or more compact directions only if nearly all such states do not contribute to the scattering amplitude. In addition, we show how the $\mathrm{d}\log$ of both open and closed string amplitudes admit a positive geometry. This allows us to give a fully geometric interpretation for the KLT double copy at four points.
\end{abstract}

\tableofcontents

\section{Introduction and Summary}
In recent years, positive geometries \cite{Arkani-Hamed:2017tmz} have emerged as a powerful framework for reinterpreting various physical quantities in a purely geometric way. Scattering amplitudes in various field theories, as well as correlators in cosmology, have been shown to admit such a geometric formulation. The usual philosophy is as follows: one starts by defining some geometric region in a suitable space, and then associate to it a unique logarithmic \emph{canonical form}. In this way, a differential form (or a function) is extracted directly from this geometry. Prominent examples include the amplituhedron \cite{Arkani-Hamed:2013jha} or the Arkani-Hamed--Bai--He--Yan (ABHY) associahedron \cite{Arkani-Hamed:2017mur}, where the relevant geometry lives in the space of kinematic variables, and the functions we extract are tree-level amplitudes or a loop-level integrands in quantum field theory.

A fundamental limitation of the traditional positive geometry framework is that it only allows for \emph{rational functions} to appear in the canonical form. Many quantities of interest fall outside of this class of functions, which makes any geometric interpretation opaque. Recently, however, progress has been made to overcome this obstacle. In \cite{Bartsch:2025mvy}, it was shown how the \emph{inverse Kawai--Lewellen--Tye (KLT) kernel} and various \emph{stringy scalar amplitudes} admit a positive geometric realisation. By considering an infinite sum of convex polytopes, they show how the canonical form can produce various trigonometric functions. For example, the infinite union of line segments
\begin{align*}
	\mathcal{A}_4^{\alpha'}\coloneqq \bigcup_{k\in\mathbb{Z}}\{\frac{k}{\alpha'}\leq z\leq \frac{k}{\alpha'}+c\},
\end{align*}
has a canonical form
\begin{align*}
	\Omega(\mathcal{A}_4^{\alpha'})=\mathrm{d}\log\frac{\sin(\pi\alpha' z)}{\sin(\pi\alpha'(c-z))} = \left(\frac{\pi}{\tan(\pi\alpha'z)}+\frac{\pi}{\tan(\pi\alpha'(c-z))}\right)\mathrm{d} z.
\end{align*}
The lattice structure of $\mathcal{A}_4^{\alpha'}$ captures the infinite tower of states which are present in these string-inspired functions.

Once we allow for such infinite positive geometries, it is natural to ask: \emph{what functions can be expressed using positive geometry? Can every scattering amplitude be represented in terms of positive geometries?} In this paper, we will answer these questions for the case of one-dimensional positive geometries. Using \emph{Hadamard's factorisation theorem}, we will classify the space of functions which can be obtained from any (possibly infinite) sum of line segments, which represent four-point scalar amplitudes. The relevant notion will turn out to be that of \emph{pseudogenus}, which we will define in section \ref{sec:complex-analysis}. This is a subtly different notion from the genus of a function from standard complex analysis, where the pseudogenus no longer requires absolute convergence. Infinite one-dimensional positive geometries have canonical forms that can be written as the $\mathrm{d}\log$ of functions with pseudogenus zero.  Additionally, we show that the spectrum encoded by an infinite union of line segments is consistent with the presence of a stringy tower of states or a Kaluza--Klein (KK) tower with four or more compact directions only if nearly all such states do not contribute to the scattering amplitude.

We furthermore show that the four-point open (Veneziano) and closed (Virasoro--Shapiro) string amplitudes have pseudogenus zero themselves, which means that their $\mathrm{d}\log$ admits a positive geometry given by an infinite sum of line segments. Since the same is true for the (inverse) KLT kernel, we can use this to give a completely geometric interpretation of the KLT double copy relations at four points. 

Finally, we will consider what happens in the \emph{continuum limit} of positive geometries. When the infinite sum of line segments become infinitesimally small and infinitesimally close together, the poles at the boundaries bunch together to form a branch cut. This is an important generalisation of the discrete case, and greatly increases the space of functions which we can capture in this framework. The relevant notion to consider is now the \emph{density of the geometry} $\rho(t)$. The canonical form of this density is then given by the derivative of the Cauchy transform of $\rho$. Using the \emph{Stieltjes--Perron inversion formula}, we can then check which function $\rho$ (if any) gives rise to a desired canonical form.

\section{Review of Complex Analysis}\label{sec:complex-analysis}
In representing a function \(f(z)\) as the canonical form \(\mathrm d\log f(z)=\Omega(S)\) of a one-dimensional positive geometry \(S\), the end-points of \(S\) correspond to zeros and poles of \(f\). The question of representability of such functions is therefore of the theory of (factorisations of) meromorphic functions, which forms part of the branch of complex analysis known as \emph{Nevanlinna theory}. When we generalise to continuum limits of such positive geometries, instead we obtain the \emph{Cauchy transform}. We therefore review these two subjects.

Most of the contents in this section is well known, with citations given to standard literature, except for the \emph{pseudogenus} and related results, which are (as far as the authors are aware) new to this work.

\subsection{Elements of Nevanlinna theory}\label{ssec:nevanlinna}
Nevanlinna theory concerns the distribution of values of meromorphic functions (and, in particular, of its zeros and poles). A basic reference is \cite{goldberg}.

The basic technical tool in Nevanlinna theory is the Nevanlinna characteristic, defined as follows.
\begin{definition}
Let \(f\) be a meromorphic function. The \emph{Nevanlinna characteristic} of \(f\) is
\begin{equation}
    T_f(r) = m_f(r) + N_f(r)
\end{equation}
where
\begin{equation}
    m_f(r) = \frac1{2\pi}\int_0^{2\pi}\max\{0,\log|f(r\mathrm e^{\mathrm i\theta})|\}\,\mathrm d\theta
\end{equation}
and
\begin{equation}
    N_f(r) = \left(\int_0^rn_f(t)\frac{\mathrm dt}t\right) + n_f(0)\log r,
\end{equation}
where \(n_f(r)\) is the number of poles (counted according to multiplicity) of \(f\) at \(|z|\le r\). Roughly, \(m_f(r)\) is the average absolute value of \(f\) on the circle of radius \(r\), while \(N_f(r)\) is a weighted count of poles of \(f\).
\end{definition}
Intuitively, the Nevanlinna characteristic therefore combines an estimate of how fast \(f\) grows and how many poles \(f\) has. It is far from obvious, however, that the above sum behaves reasonably --- in particular, the above definition seems to treat zeros and poles on a very different footing. Nevertheless, remarkably, zeros and poles are nearly equivalent for the Nevanlinna characteristic in the following sense.
\begin{theorem}[First fundamental theorem of Nevanlinna theory]
    For any meromorphic function \(f\) and any \(a\in\mathbb C\),
    \begin{equation}T_{1/(f-a)}(r) =T_f(r)+ \mathcal O(1),\end{equation}
    where the \(\mathcal O(1)\) constant can depend on \(f\) and \(a\) but not on \(r\).
\end{theorem}
Furthermore it is compatible with addition and multiplication in the sense that
\begin{align}
    T_{fg} &\le T_f+T_g + \mathcal O(1),&
    T_{f+g} &\le T_f+T_g + \mathcal O(1),&
    T_{af+b} &= T_f + \mathcal O(1)
\end{align}
for meromorphic \(f,g\) and constants \(a,b\in\mathbb C\).

Using the Nevanlinna characteristic, we may define the class of nicely behaved meromorphic functions as those with a finite order, where order is defined as follows:
\begin{definition}
    The \emph{order} of a meromorphic function \(f\) is
    \begin{equation}
        \operatorname{ord}f = \limsup_{r\to\infty}\frac{\max\{0,\log T_f(r)\}}{\log r}.
    \end{equation}
\end{definition}
In particular, for an entire function \(f\) (i.e.\ meromorphic function without poles), if \(f\) has order \(k\), then \(f\) grows as \(|f(z)|\sim a\exp(b|z|^k)\) for some constants \(a\) and \(b\). It follows from the properties of the Nevanlinna characteristic that \(f\) and \(1/f\) have the same order and that \(\operatorname{ord}(fg) \le \max\{\operatorname{ord} f,\operatorname{ord}g\}\) and \(\operatorname{ord}(f+g)\le\max\{\operatorname{ord}f+\operatorname{ord}g\}\).

Given a meromorphic function \(f(z)\colon\mathbb C\to\mathbb C\sqcup\{\infty\}\) on the complex plane, we wish to represent \(f(z)\) as a product of factors over its zeros and poles. For most \(f\), unfortunately, a naïve such attempt fails due to issues of convergence.
The \emph{Weierstrass canonical factor} is the function \cite[p.~55]{goldberg}
\begin{equation}
\begin{aligned}
    E_k(z) &= (1-z) \exp\mleft(
        z + z^2/2 + \dotsb + z^k/k
    \mright)\\
    &= \exp\mleft(
        \log(1-z)
        +z+z^2/2+\dotsb+z^k/k
    \mright)\\
    &=\exp\mleft(
        \log(1-z)
        - \left[\text{first \(k\) terms of the Taylor series of \(\log(1-z)\)}\right]
    \mright).
\end{aligned}
\end{equation}
That is, \(E_k(z)\) approximates \(1-z\) but drops the first \(k\) terms in the logarithmic Taylor series \(1-z = \exp(-z-z^2/2-z^3/3-\dotsb)\); this will improve convergence when we take a product over all zeros and poles.

Another technical concept we require is the genus of a sequence.
\begin{definition}[{}{\cite[p.~57]{goldberg}}]
    The \emph{genus} of a sequence of non-zero complex numbers \(a=(a_1,a_2,\dotsc)\) is the smallest integer \(p\) such that the sequence \(1/|a|=(|a_1|^{-1},|a_2|^{-1},\dotsc)\) belongs to \(\ell^{p+1}\), that is,
    \(
        \sum_i|a_i|^{-p-1}<\infty
    \).
\end{definition}
The higher the genus, the better behaved the sequence is in terms of convergence for the purposes of the Hadamard factorisation theorem.

With these tools, we may state the Hadamard factorisation theorem for meromorphic functions.
\begin{theorem}[Hadamard factorisation theorem {\cite[Thm.~4.1]{goldberg}}]
    Let \(f\) be a meromorphic function of finite order. Suppose that \(f(z)\) has a zero of order \(n\) at \(z=0\) (where \(n>0\) for zeros and \(n<0\) for poles and \(n=0\) otherwise).
    Let the zeros be \(a_i\) and the poles \(b_i\), counted according to multiplicity. Let \(p\) be the genus of \(a_i\); let \(q\) be the genus of \(b_i\). Then \(f\) can be canonically represented as
    \begin{equation}\label{eq:hadamard-factorisation}
        f(z) = z^n \exp(P(z))
        \prod_iE_p(z/a_i)
        \prod_i\frac1{E_q(z/b_i)},
    \end{equation}
    where \(P\) is a polynomial of degree at most \(\lfloor\operatorname{order}(f)\rfloor\).
\end{theorem}

\begin{definition}[{}{\cite[p.\ 58]{goldberg}}]
    The \emph{genus} of a meromorphic function \(f\) of finite order is the smallest non-negative integer \(g\) such that
    \begin{equation}\label{eq:genus}
        f(z) = z^n \exp(P(z))
        \prod_{a\in\operatorname{zeros}(f)}E_g(z/a)
        \prod_{b\in\operatorname{poles}(f)}\frac1{E_g(z/b)}
    \end{equation}
    for some integer \(n\) and polynomial \(P\) of degree at most \(g\); that is, \(g=\max\{k,p,q\}\) in the representation \eqref{eq:hadamard-factorisation}.
\end{definition}

\begin{theorem}[{}{\cite[p.\ 58]{goldberg}}]\label{thm:genus-order}
    For a meromorphic function \(f\) with genus \(g\) and order \(\rho\),
    \(\rho-1 \le g \le \rho\).
\end{theorem}

In the above, convergence is meant in an \emph{absolute} sense, i.e.\ it does not depend on an ordering of the zeros and poles of \(f\).
In particular, if we wish to represent \(f\) as an absolutely convergent product of factors of the form \((1-z/a)=E_0(z/a)\) only, then the necessary and sufficient condition for this to be possible is for \(f\) to have genus zero.

On the other hand, in physical applications it is often convenient to only consider non-absolute convergence, where we order the factors in \eqref{eq:genus} in a particular way (e.g.\ according to the absolute values of the zeros and poles). This motivates the following (non-standard) definition.
\begin{definition}
    Suppose \(f\) is a meromorphic function of finite order that is not identically zero.
    Let the set of poles and zeros of \(f\) (except for those at \(z=0\)) be \(a_1,a_2,\dotsc,\), \emph{not} counting multiplicity, and ordered according to increasing absolute value (with ties broken arbitrarily); let \(k_i\) be the non-zero integer such that the leading Laurent series of \(f\) at \(a_i\) is \(f(z)=c_k(z-a_i)^k+c_{k+1}(z-a_i)^{k+1}+\dotsb\) with \(c_k\ne0\). (That is, if \(k_i>0\), then \(a_i\) is a zero of order \(k_i\); if \(k_i<0\), then \(a_i\) is a pole of order \(-k_i\).)
    The \emph{pseudogenus} of a meromorphic function \(f\) is the smallest non-negative integer \(\psi\) such that
        \begin{equation}\label{eq:def-genus-rel}
            f(z) = z^n\exp(P(z)) \prod_i
            \left(E_\psi(z/a_i)\right)^{k_i},
        \end{equation}
        where \(n\) is an integer and
        where \(P(z)\) is a polynomial of degree at most \(\psi\), and \eqref{eq:def-genus-rel} is meant in the sense of uniform convergence on any compact subset of \(\mathbb C\) for the ordering of \(\{a_i\}\) according to increasing absolute values.
\end{definition}
The essential difference between the genus and the pseudogenus lies in the requirement of absolute convergence for the former.

The pseudogenus can differ from the genus by an arbitrary amount, as the following example shows.
\begin{example}
    Consider the set of zeros given by
    \begin{equation}
        a_n = \sqrt{\lceil n/4\rceil}\mathrm i^{n-1},
    \end{equation}
    such that
    \begin{equation}
       \{a_n\}= \{ 1, \mathrm i, -1, -\mathrm i, \sqrt{2}, \mathrm i\sqrt{2}, -\sqrt{2}, -\mathrm i \sqrt{2},\dotsc, \sqrt{n}, \mathrm i \sqrt{n},-\sqrt{n},-\mathrm i \sqrt{n},\dotsc\},
    \end{equation}
    according to increasing absolute value. Clearly, we have $|a_n|=\sqrt{\lceil n/4\rceil}$.
    Since $\sum_{n=1}^\infty 1/(\sqrt{n})^3<\infty$, an entire function with this set of zeros has order \(2\). We construct
    \begin{equation}
        f(z) = \prod_{n} \left(1-z/a_n\right)\exp(z/a_n + z^2/2a_n^2).
    \end{equation}
    The factor
    \begin{equation}
        \prod_n\exp(1+\frac{z}{a_n} + \frac{z^2}{2a_n^2})
        =\exp\left(\sum_n \frac{z}{a_n} + \frac{z^2}{2a_n^2}\right)
        =\exp\left(z\left(\sum_n a_n^{-1}\right) + \frac12 z^2\left(\sum_n a_n^{-2}\right)\right)
    \end{equation}
    converges non-absolutely, as it relies on the cancellations
    \begin{equation}
        \sum_n a_n^{-1} = \sum_n n^{-1}\left(1-\mathrm i-1+\mathrm i\right) = 0,
    \end{equation}
    and
    \begin{equation}
        \sum_n a_n^{-2} = \sum_n n^{-2} \left(1-1+1-1\right) = 0.
    \end{equation}
    Then,
    \begin{equation}
        f(z) = \prod_{n} \left(1-z/a_n\right)
    \end{equation}
    holds, so the pseudogenus is \(0\). More generally, by changing \(\sqrt{\lceil n/4\rceil}\) to a function slower-growing than the square root, one can engineer functions with arbitrarily large genus but with pseudogenus \(0\).
\end{example}

However, in the special case when the zeros and poles lie on the real line, the pseudogenus can differ from the genus by at most one.
\begin{lemma}\label{lemma:pseudogenus-genus-difference}
    Suppose \(f\) is a meromorphic function of finite order, with genus \(g\) and pseudogenus \(\psi\). Suppose that the poles and zeros of \(f\) lie on the real line. If \(g\) is even, then \(\psi=g\). If \(g\) is odd, then \(\psi\in\{g,g-1\}\).
\end{lemma}\begin{proof}
    Suppose that \(f\) has genus \(g\), so that
    \begin{equation}
        f(z)=z^n\exp(P(z))\prod_i(E_g(z/a_i))^{k_i}.
    \end{equation}
    Suppose that it has pseudogenus \(\psi<g\). Then
    \begin{equation}
        P(z) + \sum_{i=1}^\infty\sum_{m=\psi+1}^g\frac1m(z/a_i)^m
    \end{equation}
    must have degree \(\psi\), so that in particular we must have
    \begin{equation}\label{eq:pseudogenus-condition}
        \sum_{i=1}^\infty 1/a_i^m = 0\qquad\forall m \in\{\psi+1,\dotsc,g\}.
    \end{equation}
    Since the \(a_i\) are real, this is only possible for \(m\) odd (otherwise \(1/a_i^m>0\)). That is, \(\{\psi+1,\dotsc,g\}\) cannot contain an even number.
    If \(g\) is even, \eqref{eq:pseudogenus-condition} can never hold; if \(g\) is odd, then \eqref{eq:pseudogenus-condition} can only hold for \(m=g\), since \(g-1\) is even, in which case \(\psi=g-1\).
\end{proof}
The following is an example of a function with zeros and poles only on the real line where the genus and pseudogenus differ (by one).
\begin{example}\label{example:sin}
    Consider the entire function \(\sin z\). Its genus is \(1\) \cite[p.~197]{ahlfors}, since the set of zeros of the \(\sin\) function is \(\pi\mathbb Z\), and \(\sum_{n\in\mathbb Z\setminus\{0\}} 1/(\pi n) = \infty\), but \(\sum_{n\in\mathbb Z\setminus\{0\}} 1/(\pi n)^2 = 1/3<\infty\). On the other hand, its pseudogenus is \(0\). This can be seen from Euler's product formula
    \begin{equation}
        \sin z= z\prod_{n=1}^\infty\left(1-z^2/(\pi n)^2\right)
    \end{equation}
    (which converges absolutely and uniformly on any compact set), which can be factorised as
    \begin{equation}
        \sin z = z\prod_{n=1}^\infty\left(1+z/(\pi n)\right)\left(1-z/(\pi n)\right).
    \end{equation}
    In this formula the zeros are ordered according to increasing absolute value. However, absolute convergence fails if one can reorder the zeros in arbitrary ways, since the convergence of the above product depends on the pairwise cancellation of the two zeros \(\pm \pi n\). Indeed,
    if one considers the above product at \(z=\pi\), then
    \begin{equation}
        \sin \pi = \pi \prod_{n=1}^\infty (1+1/n)(1-1/n),
    \end{equation}
    but the partial product
    \begin{equation}
        \prod_{n=1}^\infty (1+1/n) = \frac21 \cdot \frac32 \cdot \frac43 \dotsm =\infty
    \end{equation}
    telescopes and diverges. Because the convergence of \(\prod_{n\in\mathbb Z\setminus\{0\}}(1-z/(\pi n))\) is not absolute, a reordering of the factors can produce a different result.
\end{example}

\subsection{The Cauchy Transform}\label{ssec:cauchy}
In the continuum limit, instead of \(\log f\) being a sum over the poles and zeros of \(f\), we instead wish to represent \(\log f\) as an integral over the locus of singularities of \(f\); this gives rise to the Cauchy transform, for which our main reference is \cite{akhiezer}.

\begin{definition}[{}{\cite[p.~125]{akhiezer}}]
    Given a signed measure \(\mathrm d\mu\) on the real line, the \emph{Cauchy transform} of the measure \(\mathrm d\mu\) is the function
    \begin{equation}\label{eq:cauchy-transform-defn}
        f(z) = \int_{-\infty}^\infty \frac{\mathrm d\mu(t)}{t-z}.
    \end{equation}
\end{definition}
The term `Cauchy transform' does not appear in \cite{akhiezer} but is used elsewhere in the literature, e.g.\ \cite{jurek}.\footnote{Sometimes the term `Cauchy transform' refers to a related integral where the contour of integration is not over \(\mathbb R\) but instead on the unit circle, as in \cite{cima}.}
When the measure \(\mathrm d\mu\) is supported on the positive half-line, then this transform is known under the name `Stieltjes transform', e.g. \cite[Ch.~VIII]{widder}. The name `Hilbert transform' is the variant where one instead regards \(f(z)\) as a function defined on the real line (rather than on the complement \(\mathbb C\setminus\mathbb R\) on the real line); this requires stronger assumptions on \(\mathrm d\mu\) and the use of the Cauchy principal value to render \eqref{eq:cauchy-transform-defn} well defined \cite[p.~102, Def.\ 2.5.11]{benedetto}.

Note that, although the Cauchy transform \(f\) is defined on \(\mathbb C\setminus\mathbb R\), since \(\mathrm d\mu\) is real we have
\begin{equation}\label{eq:cauchy-transform-half-plane-relation}
    f(\bar z)=\overline{f(z)}.
\end{equation}
That is, it suffices to regard \(f\) as being defined on the upper half-plane, and the values on the lower half-plane are canonically determined by \eqref{eq:cauchy-transform-half-plane-relation}. Furthermore, if the measure \(\mathrm d\mu\) is an unsigned measure, any function \(f\) representable as \eqref{eq:cauchy-transform-defn} will have non-negative imaginary component on the complex upper half-plane:
\begin{equation}
    \Im f(z) = \int_{-\infty}^\infty \Im\frac1{t-z}\,\mathrm d\mu(t)
    =\int_{-\infty}^\infty\frac{\Im z}{|t-z|^2}\,\mathrm d\mu(t)>0\text{ if }\Im z>0.
\end{equation}
Hence, a function \(f\) representable as a Cauchy transform of an unsigned measure belongs to the following class of functions.
\begin{definition}[{}{\cite[p.~92]{akhiezer}, \cite{aronszajn}}]
    A \emph{Nevanlinna function} \(f\colon\{z\in\mathbb C|\operatorname{Im}z>0\}\to\{z\in\mathbb C|\operatorname{Im}z\ge0\}\) is a holomorphic function on the complex upper (open) half-plane whose imaginary part is non-negative everywhere.
\end{definition}
(Such functions are called `functions of class \(P\)' in \cite{aronszajn}.) Nevanlinna functions admit the following integral representation.
\begin{theorem}[Herglotz representation theorem {}{\cite[p.~92]{akhiezer}, \cite{aronszajn}}]
    A function on the complex upper half-plane is Nevanlinna if and only if it is representable as
    \begin{equation}
        f(z) = C+Dz+\int_{-\infty}^\infty\left(\frac1{t-z}-\frac t{1+t^2}\right)\,\mathrm d\mu(t),
    \end{equation}
    where \(C\) is a real number, \(D\) is a non-negative real number, and \(\mathrm d\mu\) is an (unsigned) Borel measure on \(\mathbb R\) such that
    \begin{equation}\label{eq:nevanlinna-measure-finiteness}
        \int_{-\infty}^\infty\frac{\mathrm d\mu(t)}{1+t^2}<\infty.
    \end{equation}
    Furthermore, this representation is unique; \(C\), \(D\), and \(\mathrm d\mu\) can be recovered from \(f\) as
    \begin{align}
        C &= \Re f(\mathrm i), \\
        D &= \lim_{t\to\infty}\frac{f(t\exp(\mathrm i\theta))}{t\exp(\mathrm i\theta)}\qquad\forall \theta\in(0,\pi),\\
        \mathrm d\mu\left((a,b]\right) &= \lim_{\delta\to0}\frac1\pi\int_{a+\delta}^{b+\delta}\Im f(t+\mathrm i0^+)\,\mathrm dt,\label{eq:stieltjes-perron}
    \end{align}
    where \(f(t+\mathrm i0^+)=\lim_{\epsilon\to0}f(t+\mathrm i\epsilon)\).
\end{theorem}
As a corollary, a Nevanlinna function is representable as \eqref{eq:cauchy-transform-defn} if and only if
\begin{align}
    \lim_{t\to\infty}\frac{f(t\exp(\mathrm i\theta))}{t\exp(\mathrm i\theta)} &= 0,\label{eq:cauchy-cond1}\\
    \int_{-\infty}^\infty\frac t{\pi(1+t^2)}\,\mathrm d(\Im f(t+\mathrm i0^+)) &= \Re f(\mathrm i).\label{eq:cauchy-cond2}
\end{align}
In this case, the formula \eqref{eq:stieltjes-perron} is known as the Stieltjes--Perron inversion formula \cite{stieltjes,perron}, cf.\ \cite[p.~125]{akhiezer}.
\begin{example}
    The function \(f(z)=\log z\) is Nevanlinna, whose integral representation is \cite[eq.\ (1.9g)]{aronszajn}
    \begin{equation}
        f(z) = \int_{-\infty}^0\left(\frac1{t-z}-\frac t{1+t^2}\right)\,\mathrm dt.
    \end{equation}
    However, it cannot be represented as a Cauchy transform because, although \eqref{eq:cauchy-cond1} holds, \eqref{eq:cauchy-cond2} fails. That is, although \(\Im f=\operatorname H(-x)\) (here \(\operatorname H\) is the step function) defines a measure \(\mathrm d\mu\) satisfying \eqref{eq:nevanlinna-measure-finiteness} (namely, the characteristic measure of the negative real half-line), the expression \(\int_{-\infty}^\infty (t/(1+t^2))\mathrm d\mu(t)\) fails to converge.
\end{example}

\begin{example}
    The function \(f(z)=\log\mleft((z-b)/(z-a)\mright)\) is Nevanlinna if \(a<b\). Its integral representation is \cite[eq.\ (1.9k)]{aronszajn}.
    \begin{equation}
        f(z) = \int_a^b\frac{\mathrm dt}{t-z}.
    \end{equation}
\end{example}

\paragraph{Signed measures.}
The above considered the Cauchy transform of an \emph{unsigned} measure. A function \(f\) is holomorphic on the upper half-plane if and only if we can represent it as a difference of two Nevanlinna functions, both of which then satisfy \eqref{eq:cauchy-cond1} and \eqref{eq:cauchy-cond2}. In this case, for \eqref{eq:cauchy-transform-defn} to exist in the ordinary sense, the Borel measure \(\mathrm d\mu(t)=\pi^{-1}\Im f(t+\mathrm i0^+)\) must be such that
\begin{equation}
    \int_{-\infty}^\infty\frac t{1+t^2}\,|\mathrm d\mu(t)|<\infty,
\end{equation}
which reduces to
\begin{equation}
    \int_{-\infty}^\infty|\mathrm d\mu(t)|<\infty.
\end{equation}

Similar to the distinction between the genus and the pseudogenus, it is instead convenient to relax the notion of integral in \eqref{eq:cauchy-transform-defn} to only exist in the sense of the Cauchy principal value \(\mathrm{p.v.}\int_{-\infty}^\infty\), i.e.\ as the limit
\begin{equation}
    f(z) = \mathrm{p.v.}\int_{-\infty}^\infty\frac{\mathrm d\mu(t)}{t-z} = \lim_{a\to\infty}\int_{-a}^a\frac{\mathrm d\mu(t)}{t-z}.
\end{equation}
In this case, the necessary condition for the existence of the inverse Cauchy transform is weakened to
\begin{equation}\label{eq:inverse-cauchy-necessary-condition}
    \operatorname{p.v.}\int_{-\infty}^\infty\frac t{1+t^2}\,\mathrm d\mu(t)<\infty
\end{equation}
for \(\mathrm d\mu(t)=\pi^{-1}\Im f(t+\mathrm i0^+)\).

\section{Functions Corresponding to One-Dimensional Positive Geometries}
We now use Nevanlinna theory and the Hadamard factorisation theorem, as discussed in \cref{ssec:nevanlinna},
to constrain which functions can correspond to one-dimensional positive geometries. The fundamental building block is the interval $[a,b]$, which has the canonical form\footnote{In the case that $a=0$, we simply replace $1-x/a$ in the logarithm with $x$.}
\begin{align}
	\Omega([a,b])=\mathrm{d}\log\frac{1-x/a}{1-x/b}=\frac{(a-b)\mathrm{d}x}{(a-x)(b-x)}.
\end{align}
A general one-dimensional positive geometry is simply a union of intervals, and the canonical form distributes additively. We consider the one-dimensional positive geometry
\begin{align}\label{eq:stringy-A}
	\mathcal A = \bigcup_{i=0}^\infty [a_i,b_i].
\end{align}
Assuming that none of the \(a_i\) and \(b_i\) are zero, the corresponding canonical form is
\begin{equation}\label{eq:Omega-product}
    \Omega(\mathcal{A}) = \sum_{i=0}^\infty \mathrm{d}\log\frac{1-x/b_i}{1-x/a_i} = \mathrm d\log\prod_{i=0}^\infty\frac{E_0(x/b_i)}{E_0(x/a_i)}.
\end{equation}
That is, this takes the form of a product of factors of the form \eqref{eq:genus} with \(k=0\). If this product is interpreted as absolute convergence, then this implies that \(f\) must have genus zero. If, however, this infinite product is interpreted only in the sense when \(a_i\), \(b_i\) are ordered according to absolute value, then this implies that \(f\) must have pseudogenus zero (and the genus must be zero or one, according to \cref{lemma:pseudogenus-genus-difference}).

For instance, a differential form as simple as
\begin{equation}
    \Omega = \mathrm d\log(\exp(\frac12x^2)) = x\,\mathrm dx
\end{equation}
\emph{cannot} arise as the canonical form of a positive geometry, since \(\exp(\frac12x^2)\) has genus and pseudogenus two.

The requirement that \(f\) has genus at most one has strong consequences: we have
\begin{equation}
    \sum_i |a_i|^{-2}+|b_i|^{-2} < \infty.
\end{equation}
In particular, the set of poles and zeros \(\{a_1,b_1,a_2,b_2,\dotsc\}\) cannot be very dense. This rules out several kinds of towers of states that couple to massless particles.

More generally, we may consider \emph{weighted} one-dimensional positive geometries, which we define as a piecewise-constant integer-valued function \(\theta\colon\mathbb R\to\mathbb Z\).\footnote{Or, more precisely, equivalence classes of such functions, where we do not care about the value of the function \(f(b)\) at the boundary between two intervals \((a,b)\) and \((b,c)\), where the restrictions \(f|_{(a,b)}\) and \(f_{(c,d)}\) are constant.} Any such function \(\theta\) can be realised as a linear combination of the form
\begin{equation}
    \theta(x) = \sum_i c_i \operatorname H(x-a_i) + d_i \operatorname H(b_i-x)
\end{equation}
for integers \(c_i,d_i\) and real constants \(a_i,b_i\), where \(\operatorname H\) is the Heaviside step function. This is to be thought of as a formal linear combination of half-lines
\begin{equation}
    \mathcal A = \sum_i c_i [a_i,\infty) + d_i (-\infty,b_i].
\end{equation}
A finite interval \([a,b]\) is then represented as the formal linear combination \([a,\infty)-(-\infty,b]\).
The canonical form of a weighted one-dimensional positive geometry is then simply the weighted sums of the canonical forms \(\Omega_{[a_i,\infty)}=\mathrm d\log(x-a_i)\) and \(\Omega_{(-\infty,b_i]}=\mathrm d\log(b_i-x)\), i.e.
\begin{equation}
    \Omega_{\mathcal A} = \sum_i c_i \mathrm d\log(x-a_i) - d_i \mathrm d\log(b_i-x)
    = \mathrm d\log\prod_i \frac{(x-a_i)^{c_i}}{(b_i-x)^{d_i}}.
\end{equation}
Thus, the canonical form is always the exterior derivative of the logarithm of a meromorphic function of pseudogenus zero.

We conclude that the $\mathrm{d}\log$ of a function \(f\colon\mathbb C\to\mathbb C\cup\{\infty\}\) is representable in terms of a weighted one-dimensional positive geometry if and only if $f$ has pseudogenus zero.

For example, consider the function
\begin{equation}\label{eq:stringy-phi3}
    f(x) = \frac{\sin(\pi\alpha'x)}{\sin(\pi\alpha'(c-x))},
\end{equation}
which encodes the colour-ordered amplitudes of the stringy \(\operatorname{Tr}(\phi^3)\) theory considered in \cite{Bartsch:2025mvy}.
It is a meromorphic function of genus one and pseudogenus zero. It can therefore be encoded as a weighted one-dimensional positive geometry --- in fact, an ordinary one-dimensional positive geometry, namely the union of intervals
\begin{equation}
    \mathcal A = [0,c] + \frac1{\alpha'}\mathbb Z,
\end{equation}
where \(+\) denotes Minkowski sum.

The genus and pseudogenus of a function can be understood in terms of its behaviour at infinity. The Hadamard exponent $e^{P(x)}$ contributes only at infinity: $\mathrm{d}\log (e^P f) = \mathrm{d}\log f + P'\mathrm{d}x$, and the second term does not carry any residue at finite values of $x$. In coordinates $u=1/x$, we have $x^k\mathrm{d}x=-u^{-k-2}\mathrm{d}u$, so a polynomial of degree $\deg P =m\geq 1$ will induce a pole of order $m+1\geq2$ at $u=0$. For functions with (pseudo)genus zero, $P$ has to be a constant, so this Hadamard exponent does not contribute to the pole at infinity. This means that the full growth is controlled by the product over poles and zeroes. The density of states constrains this product, as we will discuss in section \ref{sec:ramification}. We stress that neither notion fixes this growth on its own, as this additionally depends on how the zeroes and poles pair up.

If $\Omega$ has a simple pole at infinity, $\Omega\simeq \alpha \mathrm{d}x/x$, then $f\simeq x^{\alpha}$. The physical interpretation of this is that $f$ must be \emph{Regge bounded} with trajectory $\alpha$. Because the tower is infinite, $x=\infty$ is an accumulation point of poles rather than an isolated singularity. This means that $\alpha$ is an asymptotic coefficient rather than a residue, and hence does not need to be an integer. This also means that the asymptotics may be direction dependent, as can be seen in example \ref{example:sin}, where $\cot( x) \to \mp i$ as $\Im(x) \to \pm\infty$, so $\sin (x)$ grows exponentially off the real axis. Nevertheless, the contributions to $\sum_i 1/a_i$ from $\pm \pi n$ cancel pairwise, which still gives it pseudogenus zero. Geometrically, this is a staircase of half-lines, where no pole of $\Omega$ has an accompanying zero. For contrast, a bounded segment $[a,b]$ adds both a pole and a zero, contributing $(b-a)/a b \sim \ell/a^2$, with $\ell=b-a$, which is suppressed by an additional power of $|a|$. Summing over a tower of segments with bounded length converges, leaving at most a simple pole at infinity. Regge boundedness in all complex directions thus requires that poles and zeros balance asymptotically, which is automatic for a union of bounded line segments.

\subsection{Ramification for Towers of States}\label{sec:ramification}
The boundaries of a positive geometry correspond to poles of the corresponding scattering amplitude. The infinite tower of boundaries present in \textit{e.g.} \eqref{eq:stringy-A} indicates a tower of states, such as those appearing in string theory or Kaluza--Klein towers.
The condition that the pseudogenus be zero (and hence that the genus be at most one) imposes strong constraints on towers of states that can contribute to the four-point amplitude.

The fact that \(p\le1\) means that
\begin{equation}
    \sum_{a\in\operatorname{zeros}(f)\setminus\{0\}}1/|a|^2 < \infty.
\end{equation}
That is, the density of states that contribute to the amplitude cannot blow up too quickly. For a uniform spacing of zeroes, such as for the $\sin$ function, we have
\begin{equation}
    \sum_{n\in\mathbb Z\setminus\{0\}}1/n^2 < \infty,
\end{equation}
so this is satisfied. Let us make a continuum approximation and say that the density of zeros is \(\sigma(x)\). Then we require
\begin{equation}
    \int_1^\infty\sigma(x)x^{-2}<\infty.
\end{equation}
This means \(\sigma(x)\sim x\) is ruled out; however, \(\sigma(x)\sim x^{1-\varepsilon}\) is fine for arbitrarily small positive \(\varepsilon\).

Now, for Kaluza--Klein towers with \(n\) compact toroidal directions, we have
\begin{equation}\sigma(x)=O(x^{n/2-1}).\end{equation}
For the stringy analogue to \(\operatorname{Tr}(\phi^3)\) described by \eqref{eq:stringy-phi3}, we have
\(\sigma(x)=O(1)\), which is the same density of states as a Kaluza--Klein tower with two compact toroidal directions.
On the other hand, for the towers which appear in string theory, one has a Hagedorn behaviour,
\begin{equation}
    \sigma(x)\sim\exp(x/T)
\end{equation}
where \(T\) is the Hagedorn temperature.

Therefore, for stringy towers or Kaluza--Klein towers with sufficiently many compact directions, the vast majority of states cannot contribute to any four-point amplitude representable as a weighted positive geometry.

\section{Geometry of Four-Point String Amplitudes}\label{sec:string}
In this section we will consider a geometric description of the four-point tree-level scattering amplitude of open-string tachyons (the Veneziano amplitude) or of closed-string tachyons (the Virasoro--Shapiro amplitude) in bosonic string theory. Although neither admits a one-dimensional positive geometry in the sense described above, we will show that they have a pseudogenus of zero, and hence that their $\mathrm{d}\log$ is the canonical form of an infinite (weighted) positive geometry. There are two interpretations of this result. First, we can interpret this as finding a positive geometry for the integrand of the logarithm of the four-point string amplitudes\footnote{Although similar in wording, this is a distinct notion from the \emph{negative geometries} considered in \cite{Arkani-Hamed:2021iya}, which capture \emph{loop} integrands for the logarithm of the four-point amplitude in $\mathcal{N}=4$ Supersymmetric Yang-Mills Theory.}. Second, these geometries directly encode properties of string scattering: the poles \emph{and} zeros of the amplitude both appear as boundaries of the positive geometry, distinguished by the sign of the residue. Additionally, we will see that the $\mathrm{d}\log$ for the KLT kernel also admits a positive geometry, which allows us to give a fully geometric interpretation to the KLT double copy relations between open and closed strings at four points.

\subsection{Veneziano Amplitude}
A tree-level colour-ordered four-point amplitude of generic open-string states has the form
\begin{equation}
\int_0^1\mathrm dt\,t^{k_3\cdot k_4/4+n_1}(1-t)^{k_2\cdot k_3/4+n_2}
\sim
\operatorname B(x,y) \equiv \mathcal{A}^{\text{open}}(x,y),
\end{equation}
for appropriate \(x\) and \(y\), where \(\operatorname B(x,y)=\Gamma(x)\Gamma(y)/\Gamma(x+y)\) is the beta function and where we have ignored colour factors and polarisations. For instance, in the simplest case, for four open-string tachyons we obtain the Veneziano amplitude
\begin{equation}
\operatorname B(-1-\alpha's,-1-\alpha't),
\end{equation}
with \(s=(p_1+p_2)^2\) and \(t=(p_2+p_3)^2\) the Mandelstam variables. It will thus suffice to consider the geometry of the beta function.

The gamma function has genus and pseudogenus one and therefore is not a product of \(E_0\) factors. Similarly, if one considers (for instance) \(\operatorname B(z,c-z)\) as a univariate function of \(z\) with a fixed \(c\), for a generic value of \(c\), this has genus one. However, it is a minor miracle that the gamma function factors combine in such a way that \(\operatorname B(z,c-z)\) in fact has pseudogenus zero rather than one, so that the Veneziano amplitude is amenable to a positive-geometry representation despite appearances.

Concretely, we have the following Hadamard representation of the full bivariate beta function, which is easy to derive from the Hadamard representation of the gamma function:
\begin{equation}\label{eq:Beta-product}
	\operatorname B(x,y)=\frac{x+y}{xy}\prod_{n=1}^\infty
	\frac{1+(x+y)/n}{(1+x/n)(1+y/n)},
\end{equation}
which is nothing more than a product over the poles and zeros of the beta function (namely, at \(x,y,x+y\in-\mathbb N\), where \(\mathbb N=\{0,1,2,\dotsc\}\)). Notably, while the Hadamard representation of the gamma function uses \(E_1\) (because it has genus and pseudogenus one), the beta function ends up having genus one but pseudogenus zero due to non-trivial cancellations, so that it can be represented using \(E_0\) instead.

Concretely, this means that we can interpret the $\mathrm{d}\log$ of the Veneziano amplitude as the canonical form of an infinite (weighted) positive geometry. Let us consider $\mathcal{A}^{\text{open}}$ as a function of $x$, while keeping $y$ some fixed constant. By expanding the $\mathrm{d}\log$ of equation \eqref{eq:Beta-product}, we find that
\begin{align}\label{eq:open-geometry}
	\mathrm{d}\log \mathcal{A}^{\text{open}}=(\psi(x)-\psi(x+y))\mathrm{d}x=\sum_{n=-\infty}^0 \Omega([n-y,n]),
\end{align}
where $\psi(x)=\partial_x \log\Gamma(x)$ is the digamma function. For values of $y$ between $0$ and $1$, this is a disjoint sum of line segments, whereas we need to consider the regions as weighted positive geometries for $y>1$.

The behaviour as $x\to\infty$ can be seen directly: $\psi(x)-\psi(x+y)\to -y/x$. We see that, asymptotically, $\mathcal{A}^{\text{open}}\sim \Gamma(y) x^{-y}$. With $x=-1-\alpha's$ and $y=-1-\alpha't$ this is the expected Regge behaviour
\begin{align}
	\mathcal{A}^{\text{open}} \sim s^{1+\alpha't}.
\end{align}
Note that each individual line-segment contributes $\Omega([a,b])\sim (a-b)\mathrm{d}x/x^2$ and is regular at infinity, so the Regge behaviour only arises collectively from the infinite tower.

\subsection{Virasoro--Shapiro Amplitude}
For the closed bosonic string, the four-point amplitude of generic external states involves the integral \cite[eq.\ (7.2.42)]{Green:2012oqa} \cite[eq.\ (6.6.22)]{Polchinski:1998rq}
\begin{multline}
	\int_{\mathbb C}\mathrm d^2z\,
	z^{a-1+m_1}\bar z^{a-1+n-1}
	(1-z)^{b-1+m_2}
	(1-\bar z)^{b-1+n_2}\\
	=
	2\sin(\pi(b+n_2))
	\operatorname B(a+m_1,b+m_2)
	\operatorname B(b+n_2,1-a-b-n_1-n_2)
\end{multline}
where \(m_1-n_1,m_2-n_2\in\mathbb Z\). Since this is a product of three factors each with genus one but pseudogenus zero, the $\mathrm{d}\log$ of the closed string amplitude can again be written as a union of infinite weighted positive geometries.

For example, when \(m_1=m_2=n_1=n_2=0\), this corresponds to the Virasoro--Shapiro amplitude, given (up to coupling constants, the Dirac delta, etc.) by \cite[eq.\ (7.2.33)]{Green:2012oqa} \cite[eq.\ (6.6.11)]{Polchinski:1998rq}
\begin{equation}
\Gamma(x,y,z)=\frac{\Gamma(x)\Gamma(y)\Gamma(z)}{\Gamma(x+y)\Gamma(y+z)\Gamma(z+x)}
\end{equation}
with \(x+y+z=1\). Using the Hadamard product
\begin{equation}
\Gamma(z) = z^{-1}\exp(-\gamma z)\prod_{n=1}^\infty (1+z/n)^{-1}\exp(z/n),
\end{equation}
we have
\begin{align}\label{eq:closed-string-product}
	\MoveEqLeft
	\mathcal{A}^{\text{closed}}(x,y)=\Gamma(x,y,1-x-y)\notag\\
	&=\notag\exp(\gamma)
	\frac{(x+y)(1-x)(1-y)}{xy(1-x-y)}\\
	&\qquad\times\prod_{n=1}^\infty
	\frac{(1+(x+y)/n)(1-x/(n+1))(1-y/(n+1))}{(1+x/n)(1+y/n)(1-(x+y)/(n+1))}
	\cdot\frac{1+1/n}{\exp(1/n)}\\
	&=\notag
	\frac{(x+y)(1-x)(1-y)}{xy(1-x-y)}
	\prod_{n=1}^\infty
	\frac{(1+(x+y)/n)(1-x/(n+1))(1-y/(n+1))}{(1+x/n)(1+y/n)(1-(x+y)/(n+1))}.
\end{align}
This factorisation now reveals the poles and zeros of \(f\): they lie at \(x,y,z,x+y,y+z,z+x\in-\mathbb N\). We simply get a product over the poles.

If we now consider $y$ to be a constant and take the $\mathrm{d}\log$ of the closed string amplitude, we find that it can be written as
\begin{align}
	\mathrm{d}\log \mathcal{A}^{\text{closed}} = \sum_{n=-\infty}^0 \Omega([n-y,n]) - \sum_{n=1}^\infty \Omega([n-y,n]).
\end{align}
That is, the $\mathrm{d}\log$ of the closed string amplitude is a positive geometry given by the infinite union of line segments $[n-y,n], \ n\in\mathbb{Z}$, but with a flipped orientation for $n\geq 1$.

\subsection{KLT Double Copy}
The open and closed string amplitudes are related through the \emph{KLT double copy}. For four-point amplitudes, we can define the inverse KLT kernel
\begin{align}
	\mathcal{A}^{\text{KLT}}(x,y) = \frac{\pi}{\tan(\pi x)}+\frac{\pi}{\tan(\pi y)} =\frac{\pi \sin(\pi(x+y))}{\sin(\pi x)\sin(\pi y)}.
\end{align}
The KLT double copy is then equivalent to the statements that
\begin{align}\label{eq:KLT}
	\mathcal{A}^{\text{closed}}= \frac{(\mathcal{A}^{\text{open}})^2}{\mathcal{A}^{\text{KLT}}},
\end{align}
which can easily be verified as a functional identity by using the reflection formula $\Gamma(1-x)\Gamma(x) = \pi/\sin(\pi x)$. 

It was shown in \cite{Bartsch:2025mvy} that the inverse KLT kernel can be given an interpretation as an infinite positive geometry:
\begin{align}
	\mathcal{A}^{\text{KLT}}(x,c-x)\mathrm{d}x = \sum_{n=-\infty}^\infty \Omega([n,n+c]).
\end{align}
However, to geometrise the KLT relations, it is important to note that $\mathrm{d}\log \mathcal{A}^{\text{KLT}}$ \emph{also} admits a positive geometry! Explicitly, we have
\begin{align}
	\mathrm{d}\log\mathcal{A}^{\text{KLT}}(x,y) = \sum_{n=-\infty}^{\infty} \Omega([n-y,n]) = \mathcal{A}^{\text{KLT}}(-x,x+y)\mathrm{d}x.
\end{align}
The function $\mathcal{A}^{\text{KLT}}(-x,x+y)$ can be interpreted as the inverse KLT kernel for a different colour ordering. The fact that the inverse KLT kernel $\mathcal{A}^{\text{KLT}}$ maps to a different colour ordering under the $\mathrm{d}\log$ map is a highly non-trivial relation, and it is far from obvious that the $\mathrm{d}\log$ of a sum of line segments can once again be interpreted as a sum of line segments.

\begin{figure}
	\centering
	\includegraphics[width=\textwidth]{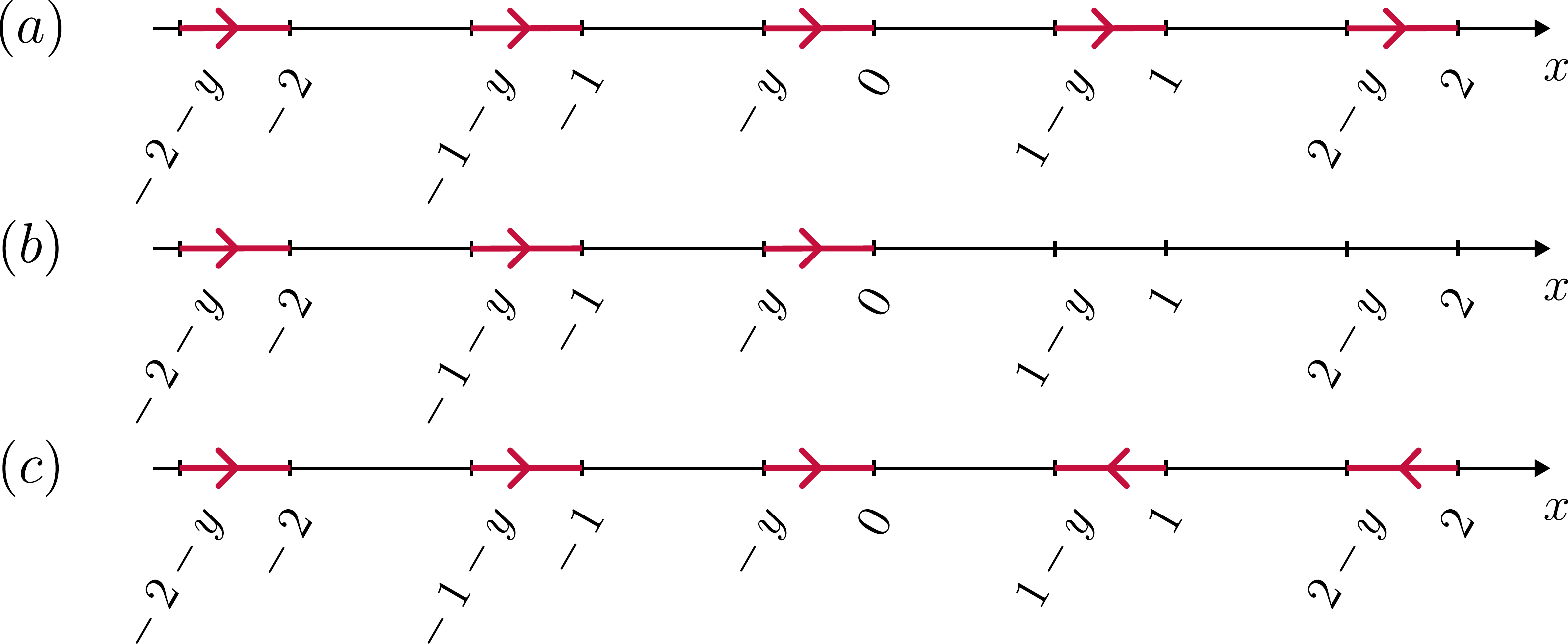}
	\caption{The infinite geometries for four-point string amplitudes: (a) $\mathrm{d}\log\mathcal{A}^{\text{KLT}}(x,y)$ (which is equivalent to $\mathcal{A}^{\text{KLT}}(-x,x+y)\mathrm{d}x$), (b) $\mathrm{d}\log \mathcal{A}^{\text{open}}(x,y)$, (c) $\mathrm{d}\log\mathcal{A}^{\text{closed}}(x,y).$}
	\label{fig:4pt-string}
\end{figure}

We can now interpret the KLT double copy as a triangulation of positive geometries. Taking the $\mathrm{d}\log$ of equation \eqref{eq:KLT}, we get
\begin{align}
	\mathrm{d}\log\mathcal{A}^{\text{closed}} &= 2 \mathrm{d}\log\mathcal{A}^{\text{open}}-\mathrm{d}\log\mathcal{A}^{\text{KLT}}\\
	\sum_{n=-\infty}^0 \Omega([n-y,n]) - \sum_{n=1}^\infty \Omega([n-y,n]) &= 2 \sum_{n=-\infty}^0 \Omega([n-y,n]) - \sum_{n=-\infty}^{\infty} \Omega([n-y,n]).	
\end{align} 
We illustrate this triangulation in figure \ref{fig:4pt-string}. Starting from the inverse KLT kernel $\mathrm{d}\log\mathcal{A}^{\text{KLT}}$, we can arrive at the open string $\mathrm{d}\log\mathcal{A}^{\text{open}}$ by discarding all line segments corresponding to $n\geq 1$, and we can get to the closed string $\mathrm{d}\log\mathcal{A}^{\text{closed}}$ from $\mathrm{d}\log\mathcal{A}^{\text{KLT}}$ by flipping the orientation for all $n \geq 1$.

\section{Continuum Limit of One-Dimensional Positive Geometries}
We will now consider the limit where the infinite line segments become arbitrarily close together and form a continuum and argue that this corresponds to the Cauchy transform as reviewed in \cref{ssec:cauchy}.
Since the end-points of the line segment correspond to the poles of the canonical form, these also get infinitesimally close together and will form a branch cut. This thus allows us to use the positive geometry framework beyond meromorphic functions, potentially incorporating examples such as (integrated) loop amplitudes.

We will let $\rho(t)$ denote the density of line segments. To be explicit, we can define
\begin{align}
	a_k=k\Delta,\quad b_k=k\Delta + \rho(k\Delta)\Delta,
\end{align}
such that
\begin{align}
	\Omega([a_k,b_k])=\frac{\rho(k\Delta)\Delta}{(x-k\Delta)(x-k\Delta-\rho(k\Delta)\Delta)}\mathrm{d}x. 
\end{align}
Taking the sum over $k\in\mathbb{Z}$, and taking the $\Delta\to0$ limit, we get
\begin{align}
	\lim_{\Delta\to0}\sum_{k\in\mathbb{Z}} \Omega([a_k,b_k])= \lim_{\Delta\to 0} \sum_k \frac{\rho(k\Delta)\Delta}{(x-k\Delta)^2} = \int \frac{\rho(t)}{(x-t)^2}\mathrm{d}t.
\end{align}
Thus, in the continuum limit of positive geometries, the canonical form is given by an integral over the density $\rho(t)$ convolved with a double pole. For example, if we want to recover a single line segment $[a,b]$, we take $\rho$ to be the indicator function
\begin{align}
	\rho(t)=\begin{cases}
		1\quad &a\leq t \leq b\\ 0 &\text{otherwise}
	\end{cases} \implies \Omega(\rho)=\int_{-\infty}^\infty \frac{\rho(t)}{(x-t)^2}\mathrm{d}t=\frac{b-a}{(x-a)(x-b)}.
\end{align}
Similarly, if we want to find the canonical form of an infinite but discrete union of line segments, then we can also simply take $\rho$ to be the indicator function. One example is a periodic sum of line segments where $\rho$ is given by a square wave:
\begin{align}
	\rho(t)=\begin{cases}
		1\quad & \lfloor t \rfloor\text{ odd} \\ 0 &\text{otherwise}
	\end{cases} \implies \Omega(\rho)=\int_{-\infty}^\infty \frac{\rho(t)}{(x-t)^2}\mathrm{d}t=\frac{\pi}{\sin(\pi x)}.
\end{align}
If we insist that the line segments in the infinite sum be all disjoint, then this requires $a_{k+1}-b_k\geq 0$. It is easy to see that this imposes the requirement that $0\leq \rho(t) \leq 1$. This is a natural requirement to consider, but it is not necessary for the integral to converge.

We note that we can define an analogous quantity
\begin{align}
	\Phi(x) = \int_{-\infty}^\infty \frac{\rho(t)}{x-t}\mathrm{d}t.
\end{align}  
This $\Phi$ can now be interpreted as the Cauchy transform of $\rho(t)$.

The result of this transform will have a branch cut along the support of $\rho$. The relevance of $\Phi$ that its derivative gives the canonical form $\Omega(\rho)$. For example, for the line segment $[a,b]$, we have
\begin{align}
	\Phi(x)=\log\frac{x-a}{x-b}\implies \partial_x \Phi = \frac{b-a}{(x-a)(x-b)} = \Omega([a,b]).
\end{align}
This allows us to reverse engineer the density $\rho(t)$ for a desired function $\Phi(x)$ using the Stieltjes--Perron inversion formula \eqref{eq:stieltjes-perron}.
\begin{align}
	\rho(x) = \lim_{\epsilon\to 0^+} \frac{\Phi(x+\mathrm i\epsilon)-\Phi(x-\mathrm i\epsilon)}{2\pi\mathrm i}.
\end{align}
This is essentially just the discontinuity of $\Phi(x)$ around the branch cut along the real axis. 

Let us look at a few simple examples. In general, a function $\Omega$ we can obtain through such an integral will have a branch cut along the subset of $\mathbb{R}$ where $\rho'\neq 0$. (If this is at isolated points, like for a step function, then $\Omega$ has a pole there.) A natural smooth version of the square wave is given by
\begin{align}
	\rho(t)=\frac{1+\sin(t)}{2}\implies \Omega(\rho) =\frac{\cos(x)(\log(x)-\log(-x))-\pi\sin(x)}{2},
\end{align}
which has a branch cut along the entire real axis. 
If we only want a branch cut along some subset of the real axis, then we can take $\rho$ such that $\rho'$ only has support on only that subset. For example,
\begin{align}
	\rho(t)=\begin{cases}
		0\quad & t<0\\
		t &0\leq t \leq 1\\1 & t > 1
	\end{cases} \implies \Omega(\rho) = \log(1-x)-\log(x),
\end{align}
which has a branch cut along $0\leq x \leq 1$.

Now, the general discussion in \cref{ssec:cauchy}, especially the necessary condition \eqref{eq:inverse-cauchy-necessary-condition}, imposes constraints on what kinds of holomorphic functions with branch cuts on the real line can be represented in terms of a continuum weighted positive geometry. In particular, a necessary condition is \eqref{eq:inverse-cauchy-necessary-condition}. If the function \(f(z)\) to be represented is a (loop-level) scattering amplitude, then the branch cut represents a continuum of intermediate states that contribute to the amplitude, and the density \(\rho\) should be non-negative. In this case, we must have
\begin{equation}
    \int \frac t{1+t^2}\,\rho(t)\,\mathrm dt < \infty,
\end{equation}
so that in particular \(\rho(t)\) must fall off as \(\mathcal O(t^{-1-\epsilon})\) for some \(\epsilon>0\) in the large-\(t\) limit. This is analogous to the genus constraint we derived for the discrete case.

\section{Conclusion and Outlook}
In this paper we have considered a broad class of one-dimensional positive geometries, and we have given strong restrictions on the class of functions which can appear. When we consider an infinite, but discrete, sum of line segments, then we have shown that their canonical forms can be found as the logarithmic differential form of a function with \emph{pseudogenus zero}. We have argued that this density of states implies that almost all states in stringy or Kaluza--Klein towers (for four or more compact directions) do not contribute to the scattering amplitude.

Furthermore, we have given a positive geometric interpretation of both open (Veneziano) and closed (Virasoro--Shapiro) string amplitudes by considering their $\mathrm{d}\log$. This geometry encodes the location of both the poles and the zeroes of string amplitudes. In combination with similar results for the (inverse) KLT kernel, this has allowed us to give a completely geometric interpretation of the KLT double copy relations for four-point string amplitudes.

Finally, we have investigated a generalisation of these infinite positive geometries to a continuum case where poles can get arbitrarily close together. We have shown how this continuum limit introduces canonical forms with a branch cut structure. This greatly increases the class of functions which admit a positive geometry, and can potentially be important for applications for (integrated) loop amplitudes. 

In addition to placing some fundamental constraints on the space of functions which admit a positive geometry, we have simultaneously greatly expanded it beyond its traditional scope. This opens several exciting avenues for future research. The scope of the current paper was restricted to the space of one-dimensional positive geometries, and it is natural to ask whether a similar analysis is possible in the multivariate case. Unfortunately, multivariate complex analysis is significantly more complicated from an analytical viewpoint, and as such a full classification of the function space will be difficult. Nonetheless, it will be interesting to investigate the infinite and continuum geometries in higher dimensions. The most urgent question concerns the geometric description of string amplitudes beyond four points. Generalisations of our discussion in section \ref{sec:string} to higher-point string amplitudes would provide an important connection between more general string amplitudes and positive geometries. The fact that the inverse KLT kernel has a positive geometric description to any multiplicity \cite{Bartsch:2025mvy} gives an important starting point for future considerations. 

\section*{Acknowledgements}
JS is supported by OP JAK CZ.02.01.01/00/22\_008/0004632.

\newcommand\cyrillic[1]{\fontfamily{Domitian-TOsF}\selectfont \foreignlanguage{russian}{#1}}
\bibliographystyle{unsrturl}
\bibliography{biblio}

\end{document}